\documentclass[10pt]{article}
\usepackage{amsmath,amssymb,amsthm,amsfonts,xcolor,graphicx,verbatim,dsfont,pdfpages,interval,pgfplots,tikz,authoraftertitle,enumitem,enumitem,hyperref,cleveref,lipsum}
\usepackage[top=2cm, bottom=1.8cm,left=2.5cm,right=2.5cm,head=2cm,headsep=0.5cm]{geometry} 
\newtheorem{theorem}{Theorem}[section]

\newtheorem{lemma}[theorem]{Lemma}
\newtheorem{example}[theorem]{Example}
\newtheorem{defn}[theorem] {Definition}

\newtheorem{Two Examples of equivariant maps}[theorem] {Two Examples of equivariant maps}

\newtheorem{Proof of the theorem 2.3:}{Proof of the theorem 2.3:}
\newtheorem{Some results from Group representation theory} [theorem] {Some results from Group representation theory}

\newtheorem{corollary}[theorem]{Corollary}

\newtheorem{Future work} [theorem]{Future work}
\newtheorem{thm}[theorem] {Theorem}

\newtheorem{Properties of equivariant linear maps} [theorem]{Properties of equivariant linear maps}

\newtheorem{Main result:}[theorem]  {Main result:}

\newlist{Properties}{enumerate}{10}
\setlist[Properties]{label*=\arabic*}

\newcommand{\Hcal}{\mathcal{H}}
\newcommand{\Kcal}{\mathcal{K}}

\newcommand{\Acal}{\mathcal{A}}
\newcommand{\Ucal}{\mathcal{U}}
\newcommand{\Scal}{\mathcal{S}}

\newcommand{\Bcal}{\mathcal{B}}

\newcommand{\Cbb}{\mathbb{C}}

\newcommand{\NN}{\mathbb{N}}

\newcommand{\SEP}{\operatorname{SEP}}

\newcommand{\Tr}{\operatorname{Tr}\,}

\begin{document}
\title{Characterization of equivariant maps and application to entanglement detection}
\author{Ivan Bardet, Beno\^\i t Collins, Gunjan Sapra}
\maketitle
\begin{abstract}
	
We study equivariant linear maps between finite-dimensional matrix algebras, as introduced in \cite{COLLINS2018398}. These maps satisfy an 
algebraic property which makes it easy to study their positivity or $k$-positivity. They are therefore particularly suitable for applications to 
entanglement detection in quantum information theory. We characterize their Choi matrices. 
In particular, we focus on a subfamily that we call $(a,b)$-unitarily equivariant. They can be seen as both 
a generalization of maps invariant under unitary conjugation as studied by Bhat in \cite{bhat2011} and as a generalization of the equivariant 
maps studied in \cite{COLLINS2018398}. Using representation theory, we fully compute them and study their graphical representation, 
and show that they are basically enough to study all equivariant maps. We 
finally apply them to the problem of entanglement detection and prove that they form a sufficient (infinite) family of positive maps to detect all $k$-entangled density matrices.
	\end{abstract}
	
\section{Introduction}

Due to their crucial role in numerous tasks in quantum processing and quantum computation, it is of great importance to decide whether a 
certain density matrix on a bipartite system is entangled or not~\cite{nielsen2002quantum,terhal2002detecting}. However this problem, referred 
to as \emph{entanglement detection}, is known to be a computationally hard one in quantum information 
theory~\cite{Gharibian:2010:SNQ:2011350.2011361,gurvits2003classical}. In the last two decades, lots of effort have been accomplished 
in order to determine necessary and sufficient conditions for a density matrix to be entangled. For instance, one such criterion is the $k$-extendibility hierarchy~\cite{doherty2004complete}, which provides a sequence of tests to check that the density matrix is separable, that ultimately detects all entangled states.\\

Another appealing method is the positive map criterion~\cite{HORODECKI19961}, which gives an operational interpretation of the Hahn--Banach theorem applied to the convex set of separable density matrices. The Horodecki's Theorem thus states that a density matrix $\rho$ is entangled if and only if there exists a positive map $\Phi$ such that $(i\otimes\Phi)(\rho)$ is not positive semi-definite, where $\Phi$ only acts on 
one of the two subsystems. Necessarily, this positive map is not completely positive.\\

The most well-known example of such map is the transpose map, and it leads to the \emph{positive partial transpose} (PPT) 
criterion~\cite{peres1996separability}. However, because of the complex geometrical structure of the set of separable density matrices, an 
infinite number of maps that one does not know how to describe efficiently would be necessary to detect all entangled states {(see for instance 
\cite{doi:10.1063/1.4712302}). The goal of this article is to propose a family of maps, with increasing complexity, which suffice to detect 
any entanglement. Their main interest lies in that it is rather easy to check if they are positive or not.\\

Compared to the $k$-extendibility hierarchy, the Horodecki's Theorem can be generalized to test the degree of entanglement of a density matrix quantified by its Schmidt number. 
A density matrix $\rho$ has Schmidt number less than or equal to $t$ if and only if $(i\otimes\Phi)(\rho)\geq0$ for all $t$-positive maps.

Similarly to positivity, we propose a sufficient family of $t$-positive maps for the problem of $t$-positivity detection. Again, it is rather easy to check if they are $t$-positive or not.\\

Indeed, proving that a map is completely positive is easy: it is enough to check that its Choi matrix is positive 
semi-definite~\cite{choi1975completely}. Such a criterion does not exist in general to check that a given map is $t$-positive, which makes it 
difficult to find interesting examples of positive but not completely positive maps. Choi \cite{choi1072positive} gave in 1973 the first example 
of a linear map on $M_{n}(\mathbb{C})$ which is $(n-1)$-positive but not $n$-positive. One decade later in 1983, Takasaki and Tomiyama 
\cite{Takasaki1983} gave a method to construct any number of linear maps on finite-dimensional matrix algebras, which are $(k-1)$-positive and 
not $k$-positive. Interestingly, all these examples fall in the class of maps introduced by Collins et al. in~\cite{COLLINS2018398}, called 
\emph{equivariant maps}. In the same article, they proved that if a linear map happens to be equivariant, its $k$-positivity depends upon the 
positivity of a $k$-blocks submatrix of the corresponding Choi matrix. In a sense, it means that it is as easy to check that an equivariant map is 
$k$-positive, as it is to check that it is completely positive. They subsequently studied a parametric family of equivariant linear maps on 
$M_{3}(\mathbb{C})$ with values in $M_3(\Cbb)^{\otimes 2}$. In this article, we analyse in more depth equivariant maps from 
$M_{n}(\mathbb{C})$ to $M_n(\Cbb)^{\otimes k}$ for all $k,n\geq1$, and characterize a large class of them. Note however that not all known examples of positive maps are equivariant (see for instance \cite{cho1992generalized,muller2018decomposability}).\\

More precisely, we are concerned with two different objectives. The first one is to get a full understanding of equivariant maps. We only get 
sparse results in this direction. As a first insight, we give a characterization of equivariant maps in terms of their Choi matrices in Theorem 
\ref{V(U)equiv}. We also define a subclass of them, the unitarily equivariant maps, which are more tracktable objects. \Cref{coro_main} is one of 
the main results of this paper, where we prove that unitarily equivariant maps - were the equivariance property is given in terms of a unitary 
representation - are in fact covariant maps with respect to a unitary representation of the unitary group. \\

The second objective is to fully compute a subclass of such maps, the $(a,b)$-unitarily equivariant maps, 
generalizing the examples in \cite{COLLINS2018398} and the characterization in \cite{bhat2011} of linear maps invariant under conjugation. \\

We then focus on the application to entanglement detection of such maps. We prove that any $t$-entangled density matrix - that is with Schmidt Number strictly lower than $t$ - can be detected using a 
$t$-positive unitarily equivariant map. Combined with previous result, this gives an explicit family of $t$-positive maps that are sufficient of $t$-entanglement detection.\\

This article is organized as follows. In Section \ref{sect1}, we introduce the equivariant maps, list some of their properties and give some 
examples, among which the one of Choi \cite{choi1072positive}, Takasaki and Tomiyama \cite{Takasaki1983}. In section \ref{sect2}, we give 
different characterizations mentioned above. We study the graphical representations of the Choi matrices of $(a,b)$-unitarily equivariant maps in Section \ref{sect3}. We focus on entanglement detection in 
Section \ref{sect4}.

\section{Equivariant maps: definitions and examples}\label{sect1}
In this section, we present definitions of equivariant linear maps and give an explanation as to why it is important to study these maps. 

\subsection{Notations, definitions and first examples}

For a positive integer $n$, $M_n(\Cbb)$ is the set of square matrices, with entries from $\mathbb{C}$ of size $n$, with canonical 
orthonormal basis $(e_{ij})_{1\leq i,j\leq n}$. The unitary group on $\Cbb^n$ is denoted by $\Ucal_n$. We denote by $\mathds{1}_{n}$ 
the identity matrix in $M_{n}(\mathbb{C})$ and by $i_{n}$ the identity map acting on $M_{n}(\mathbb{C})$. We write 
$B_{n}= \sum_{i=1}^{n}|e_{i}\rangle \otimes |e_{i}\rangle$, the non-normalized maximally entangled \emph{Bell} vector in 
$(\mathbb{C}^{n} \otimes \mathbb{C}^{n})$. The rank-one projection on $B_n$ is denoted by $B_{n^{2}}=B_{n}B_{n}^{*}$. Finally, $A^{t}$ and 
$\mathrm{Tr}(A)$ denote the transpose and (non-normalized)
trace of a matrix $A \in M_{n}(\mathbb{C})$ respectively. $\theta_{n}$ denotes the transpose map $A \mapsto A^{t}$ on $M_{n}(\mathbb{C})$.\\

If $\Hcal$ and $\Kcal$ are two complex Hilbert spaces, $\Bcal(\mathcal{H})$ denotes the space of all bounded linear operators on 
$\mathcal{H}$ and $\Bcal(\Hcal,\Kcal)$  the space of all bounded linear maps from $\Hcal$ to $\Kcal$. A self-adjoint map 
$\Phi\in \Bcal(\Hcal,\Kcal)$ is such that $\Phi(X^*)=\Phi(X)^*$ for all $X\in \Hcal$. In the following, we will assume that all linear maps are self-adjoint.\\

We are now ready to give the main definitions of this article.

\begin{defn}Let $n,N \geq 2 $ be two natural numbers. A 
linear map $\Phi:M_{n}(\mathbb{C}) \rightarrow M_{N}(\mathbb{C})$ 
is called:
\begin{enumerate}[label=(\roman*)]
\item\emph{Equivariant}, if for every unitary matrix $U \in M_{n}(\mathbb{C})$ there exists $V = V(U) \in M_{N}(\mathbb{C})$ such that 
\begin{align} \label{eqgeneral}
\Phi(UXU^{*}) & =V(U)\, \Phi(X)\,V(U)^{*} \qquad \forall \hspace{2mm} X \in  M_{n}(\mathbb{C});
\end{align}
\item\emph{Unitarily equivariant}, if furthermore the operator $V(U)$ in the previous definition can be taken unitary.
		\item\emph{$(a,b)$-unitarily equivariant}, if there are $a, b$ natural numbers such that $N=n^{a+b}$ and $M_N(\Cbb)\equiv M_n(\Cbb)^{\otimes a}\otimes M_n(\Cbb)^{\otimes b}$, and such that for every unitary $U \in M_{n}(\mathbb{C})$,
		\begin{equation}\label{equivariant}
		\Phi(UXU^{*}) = (\overline{U}^{\otimes a} \otimes U^{\otimes b})\,\Phi(X)\, (\overline{U}^{\otimes a} \otimes U^{\otimes b})^{*} \qquad \forall \hspace{2mm} X \in M_{n}(\mathbb{C}).
		\end{equation}
\end{enumerate}
\end{defn}
Thus, $(a,b)$-unitarily equivariant maps are a subfamily of unitarily equivariant maps, that is itself a 
subclass of equivariant maps. 
We  prove in \Cref{coro_main} that every unitarily equivariant map where $U\mapsto V(U)$ is a 
unitary representation can be seen as a corner of a sum of
$(a,b)$-unitarily equivariant maps. 
In the following, we give a family of equivariant and not unitariy equivariant maps.
\begin{example}
Let $A \in M_{n}(\mathbb{C})$ be a non-unitary invertible matrix. Define:
\begin{equation*}\label{eq_example_Choi}
 \begin{array}{cccc}
  \Phi_{A} :& M_{n}(\mathbb{C}) &\rightarrow & M_{n}(\mathbb{C})  \\
        & X                & \mapsto    & AXA^{*}
 \end{array}
\end{equation*}
Let $U \in M_{n}(\mathbb{C})$ be a unitary. Then for every $X \in M_{n}(\mathbb{C})$,
\begin{align*}
\Phi(UXU^{*})& = AUXU^{*}A^{*} \\
             & = (AUA^{-1})(AXA^{*})(A^{*})^{-1}U^{*}A^{*} \\
             & = (AUA^{-1})\Phi(X)(AUA^{-1})^{*}
\end{align*}

Clearly, $AUA^{-1}$ may not be a unitary as $A$ is not unitary. Therefore, $\Phi_{A}$ is equivariant and not unitarily equivariant.

\end{example}

\begin{example}
Bhat characterized in \cite{bhat2011} all $(0,1)$-unitarily equivariant maps (that is, for $a=0$ and $b=1$) on $\Bcal(\Hcal)$ for some Hilbert space 
 $\Hcal$, not necessarily finite dimensional. More precisely, he proved that a linear map $\Phi$ acting on $\Bcal(\Hcal)$ satisfies 
 $\Phi(UXU^*)=U\Phi(X)U^*$ for all $X\in\Bcal(\Hcal)$ iff there exist $\alpha,\beta\in\Cbb$ such that
 \[\Phi(X)=\alpha X+\beta\Tr[X]I_\Hcal\,.\]
Directly from this, we get that any $(1,0)$-unitarily equivariant map $\Psi$ on $M_{n}(\mathbb{C})$
 is of the form 
$\Psi(X)=\alpha \theta_n(X)+\beta\Tr[X]I_\Hcal$, where $\theta_n$ is the transpose map. Indeed, we can check that 
$\theta_n\circ\Psi$ is $(0,1)$-unitarily equivariant and apply Bhat's result. We shall similarly characterize all 
$(a,b)$-unitarily equivariant maps in \Cref{Ueq}, thus generalizing these two case, when $\Hcal$ is 
finite dimensional.
\end{example}

Establishing $k$-positivity of a linear map is a difficult task, even on low dimensional matrix algebras. In this regard, a criterion which is 
a necessary and sufficient condition for an equivariant map to be $k$-positive was given in \cite{COLLINS2018398}. 

\begin{thm}\label{k-positivity}\cite[Theorem 2.2]{COLLINS2018398} 
	Let $\Phi:{M}_{n}(\mathbb{C}) \rightarrow M_{N}(\mathbb{C})$ be an equivariant map. Then, for $k \leq min\{n,N\},\hspace{0.5mm} \Phi$ 
	is $k$-positive if and only if the block matrix $[\Phi(e_{ij})]_{i,j=1}^{k}$
	is positive semi-definite, where $(e_{ij})_{1\leq i,j\leq n}$ are the matrix units in $M_{n}(\mathbb C).$ 
\end{thm}

Incidentally, some well-known examples of $k$-positive but not completely positive linear maps are actually examples of unitarily 
equivariant maps, even though this was not explicitly stated when they were introduced. We recall these examples and give alternative 
proof of their $k$-positivity, based on \Cref{k-positivity}.  

\begin{enumerate}[label=(\roman*)]
\item Every $*$-homomorphism or $*$-anti-homomorphism on a finite-dimensional matrix algebras is equivariant. Such maps are 
always completely positive or co-completely positive.
\item  Choi \cite[Theorem 1]{choi1072positive} gave the first example of a linear map on $M_{n}(\mathbb{C})$ which is $(n-1)$-positive 
but not $n$-positive, given by
\begin{equation*}\label{eq_example_Choi}
 \begin{array}{cccc}
  \Phi :& M_{n}(\mathbb{C}) &\rightarrow & M_{n}(\mathbb{C})  \\
        & A                 & \mapsto    & \{(n-1)\mathrm{Tr}(A)\}\mathds{1}_{n}- A\,.
 \end{array}
\end{equation*}
The above map is unitarily equivariant (take $V(U)=U$). We apply Theorem \ref{k-positivity} to prove that $\Phi$ is $(n-1)$-positive. 
The eigenvalues of $[\Phi(e_{ij})]_{i,j=1}^{n-1}$ are $0$ with multiplicity $1$ and $(n-1)$ with multiplicities $(n(n-1)-1),$ which are positive. 
Hence, $\Phi$ is $(n-1)$-positive. 
\item Tomiyama~\cite[Theorem 2]{TOMIYAMA1985169} gave an example of a parametric family of linear maps and studied the 
conditions of its $k$-positivity. The map is defined by:
\begin{equation*}\label{eq_example_Tomiyama}
 \begin{array}{cccc}
  \Psi :& M_{n}(\mathbb{C}) &\rightarrow & M_{n}(\mathbb{C})  \\
        & A                 & \mapsto    & \frac{\lambda}{n}\mathrm{Tr}(A)\mathds{1}_{n} +(1-\lambda)A\,.
 \end{array}
\end{equation*}
This map is $k$-positive for any $1 \leq k \leq n$ if and only if $0 \leq \lambda \leq 1+\frac{1}{nk-1}$. We give an alternative proof based 
on Theorem \ref{k-positivity}. Indeed, $\Psi$ is unitarily equivariant (take $V(U)= U$) and to find the conditions of $k$-positivity on $\Psi$, 
we only need to find the values of $\lambda$ such that $[\Psi(e_{ij})]_{i,j=1}^{k}$ is positive. It can be easily seen that $\frac{\lambda}{n}$ 
and $\frac{\lambda}{n}+(1-\lambda)k$ are two distinct eigenvalues of $[\Psi(e_{ij})]_{i,j=1}^{k}$ with different multiplicities. Therefore, the 
map $\Psi$ is $k$-positive if and only if $\lambda \geq 0$ and $\lambda \leq 1+\frac{1}{nk-1}$.
 
\item Collins et.al \cite{COLLINS2018398} gave the following family of parametric linear maps which are $(1,1)$-unitarily equivariant.\\
Let $\alpha$ and $\beta$ be two real numbers and $n \geq 3.$ Then the family of maps,
\begin{equation}\label{eq_Phi_alphabeta}
\begin{array}{cccc}
	\Phi_{\alpha, \beta,n} :& M_{n}(\mathbb{C}) &\rightarrow & M_{n}(\mathbb C) \otimes M_{n}(\mathbb C)  \\
	                          & A               & \mapsto &  A^{t} \otimes \mathds{1}_{n}+ \mathds{1}_{n} \otimes A + \mathrm{Tr}(A)(\alpha \mathds{1}_{n^{2}}+\beta B_{n^{2}})\,. 
	                          \end{array}
\end{equation}

There are values of parameters $\alpha$ and $\beta$ for which the family of maps $\Phi_{\alpha, \beta,n}$ is positive and not 
completely positive, more detail can be found in \cite{ COLLINS2018398}.
 \end{enumerate} 

\subsection{Properties of equivariant linear maps}
In this section, we study some basic properties of equivariant linear maps, which will help us to give a characterization of these maps.

\begin{lemma} \label{vspace}\
Let $a,b \in \mathbb{N}$. Then,
\begin{enumerate}
\item The set of all $(a,b)$-unitarily equivariant maps is a vector space.
\item The set of all unitarily equivariant maps is not a vector space. Neither is the set of equivariant maps. 
\end{enumerate}
\end{lemma}
\begin{proof}
It is an easy calculation to show that the set of all $(a,b)$-unitarily equivariant maps is closed under 
addition and scalar multiplication. Hence, this set is a vector space.\\
We give an example of two unitarily equivariant maps such that their sum is not even equivariant, which is enough to conclude that both 
sets are not vector spaces.

Let $\Phi_{1},\Phi_2:M_{2}(\mathbb{C}) \rightarrow M_{2}(\mathbb{C})$ be given by $\Phi_1=i_2$ and $\Phi_{2}=\theta_2$. It is 
straightforward that both maps are unitarily equivariant. We prove that the map $(\Phi_{1}+\Phi_{2})$ is not equivariant, that is, there exists 
a unitary $U \in M_{2}(\mathbb{C})$ such that there is no $V \in M_{2}(\mathbb{C})$ with 
$$(\Phi_{1}+\Phi_{2})(UXU^{*})  =V (\Phi_{1}+\Phi_{2})(X)V^{*} \hspace{2mm} \forall \hspace{2mm} X \in  M_{2}(\mathbb{C})\,.$$
Let \\ $$U=  
\begin{bmatrix}
\frac{1}{\sqrt{2}} & \frac{1}{\sqrt{2}}  \\
\frac{-\iota}{\sqrt{2}} & \frac{\iota}{\sqrt{2}}  \\
\end{bmatrix}\,.$$
Take $X=e_{11}$, where $e_{11}$ is a matrix unit in $M_{2}(\mathbb{C})$ and $\iota^{2}=-1$. 
We prove that there is no $V \in M_{2}(\mathbb{C})$ such that,
		\begin{align}\label{counter}
		(\Phi_{1}+\Phi_{2})(Ue_{11}U^{*}) & = V (\Phi_{1}+\Phi_{2})(e_{11})V^{*}\,.
		\end{align} \\
$$(\Phi_{1}+\Phi_{2})(Ue_{11}U^*)= \,
\begin{bmatrix}
	1 &  0  \\
	0 &  1  \\
\end{bmatrix}\,.$$
Therefore, $ rank((\Phi_{1}+\Phi_{2})(Ue_{11}U^*))= 2$.\\ Since, $rank(\Phi_{1}+\Phi_{2})(e_{11})=1$. Therefore there exists no $V \in M_{2}(\mathbb{C})$ such that,
\begin{align}\label{counter}
(\Phi_{1}+\Phi_{2})(Ue_{11}U^{*}) & = V (\Phi_{1}+\Phi_{2})(e_{11})V^{*}\,.
\end{align} 
 
	\end{proof}	

\begin{example}\normalfont
To conclude this section, we give an example of a completely positive linear map on $M_{2}(\mathbb{C})$ which is not equivariant. 
Consider the linear map 
\[\begin{array}{cccc}
   \Phi:& M_{2}(\mathbb{C}) & \rightarrow & M_{2}(\mathbb{C})   \\
        & A                 &\mapsto      & A - A^{t}+ \mathrm{Tr}(A) \mathds{1}_{2}
  \end{array}\,.\]
The map $\Phi$ is completely positive, but it is not equivariant.

Let \\ $$U=  
		\begin{bmatrix}
		\frac{1}{\sqrt{2}} & \frac{1}{\sqrt{2}}  \\
		\frac{-\iota}{\sqrt{2}} & \frac{\iota}{\sqrt{2}}  \\
		\end{bmatrix}
\text{and}  \quad X=  -\iota e_{11}+\iota e_{22}$$
Where $\iota^{2}=-1$. 
		Then 
		$$\Phi(UXU^*)=  
		\begin{bmatrix}
		0 & 2 \\
		-2 &  0  \\
		\end{bmatrix} and   \quad \Phi(X)=  
		\begin{bmatrix}
		0 & 0  \\
		0 & 0 \\
		\end{bmatrix}\,.$$
		Therefore, there exists no $V \in M_{2}(\mathbb{C})$ such that $\Phi(UXU^{*})=V\Phi(X)V^{*}$ for unitary $U$ and $X=-\iota e_{11}+\iota e_{22}.$ Hence, $\Phi$ is not equivariant.
 
\end{example}

\section{Characterization of equivariant maps}\label{sect2}
In this section, we tackle the problem of characterizing equivariant maps defined on finite-dimensional matrix algebras. We first 
characterize general equivariant linear maps in terms of their Choi matrices in \Cref{V(U)equiv}. 
Then, we prove that every unitarily 
equivariant map where $U\mapsto V(U)$ is a unitary representation can be realized as an 
$\oplus_{i\in I}(a_{i},b_{i})$-unitarily equivariant map for 
a finite set $I$ and for
some values of $a_i$ and $b_i$. Finally, we compute explicitly all 
$(a,b)$-unitarily equivariant maps in \Cref{Ueq}. 

\subsection{Characterization of equivariant maps in terms of their Choi matrix} 
We recall that the Choi matrix $C_\Phi$ of a linear map $\Phi\in\Bcal(M_n(\Cbb),M_N(\Cbb))$ is defined by 
\begin{equation} \label{Choimatrixdefn}
C_\Phi :=(i_n\otimes \Phi)(B_{n^2})
\end{equation} 
where $B_{n^2}$ is the unnormalized Bell density matrix. Equation (\ref{Choimatrixdefn}) defines a one-to-one 
linear correspondence between matrices $C\in M_{nN}(\Cbb)$ and maps $\Phi\in\Bcal(M_n(\Cbb),M_N(\Cbb))$. The following theorem is 
applicable to all equivariant maps. 

\begin{theorem}\label{V(U)equiv}
	Let $\Phi:M_{n}(\mathbb{C}) \rightarrow M_{N}(\mathbb{C})$ be a linear map. Then the following two assertions are equivalent:
	\begin{enumerate}[label=(\roman*)]
		\item  $\Phi$ is an equivariant map.
		\item For all unitary matrices $U$ in $M_n(\Cbb)$, there exists a matrix $V=V(U)\in M_N(\Cbb)$ such that $$[C_{\Phi},\overline{U}\otimes V]=0.$$
	\end{enumerate} 
\begin{proof}
To prove $(i) \implies (ii)$, we use the fact that $(\overline{U} \otimes U)$ commutes with $B_{n^{2}}$ for every unitary matrix 
$U \in M_{n}(\mathbb{C})$. We have the following
\begin{align}
C_{\Phi} &= (i_{n} \otimes \Phi)(B_{n^{2}}) \nonumber \\
& = (i_{n} \otimes \Phi)(\overline{U} \otimes U)(\sum_{i,j}^{} e_{ij} \otimes e_{ij})(\overline{U} \otimes U)^{*}  \nonumber \\
& = (\sum_{i,j}^{} \overline{U}e_{ij} \overline{U}^{*}) \otimes \Phi(Ue_{ij}U^{*})  \label{eqChoi_equiv} 
\end{align}
Since $\Phi$ is equivariant, for every unitary $U\in M_{n}(\mathbb{C})$ there exists $V \in  M_{N}(\mathbb{C})$ such that $\Phi$ 
satisfies Equation \eqref{eqgeneral}. From Equation (\ref{eqChoi_equiv}), $C_{\Phi}$ becomes,
\begin{equation}
\begin{aligned}
 C_{\Phi} & =  \sum_{i,j}^{} \overline{U} e_{ij} {\overline{U}}^{*} \otimes V \Phi(e_{ij}) V^{*}  \nonumber\\
& = (\overline{U} \otimes V)C_{\Phi}(\overline{U} \otimes V)^{*} \,. \label{eq_Choi}
\end{aligned}
\end{equation}
Hence, $C_{\Phi}$ commutes with $(\overline{U} \otimes V)$.\\
 We now prove the converse implication $(ii) \implies (i)$.	
Assume that assertion $(ii)$ holds. Then the previous computation shows that $(\overline{U} \otimes V)C_{\Phi}(\overline{U} \otimes V)^{*}$ 
is the Choi matrix of the map $Y\mapsto V\Phi(U^*YU)V^*$. By unicity of the Choi matrix, we get that for all $Y\in M_n(\Cbb)$,
\[\Phi(Y)=V\Phi(U^*YU)V^*\,.\]
Taking $X=U^*YU$ yields the result. Hence, the proof is completed.
\end{proof}
\end{theorem}

\subsection{Characterization of unitarily equivariant maps}

The definition of equivariant maps given by Equation (\ref{eqgeneral}) is quite general in the sense that we do not put any constraint on the operator 
$V(U)$. 
It turns out that in the case of unitary invariance, the following theorem implies that $U\to V(U)$ might
be assumed to be a group morphism:

\begin{theorem}\label{theo_char_unitaryequivariant}
	Let $\Phi:M_{n}(\mathbb{C}) \rightarrow M_{N}(\mathbb{C})$ be a unitarily equivariant linear map. Then there exists a 
	\emph{unitary representation} $U\mapsto V(U)$ of the unitary group $\Ucal_n$ on $\Cbb^{N}$ such that for all $X\in M_n(\Cbb)$ and 
	all unitary $U\in\Ucal_n$,
	\[\Phi(UXU^*)=V(U)\,\Phi(X)\,V(U)^*\,.\]
\end{theorem}

\begin{proof}
We denote by $\Acal(\Phi)$ the $C^*$-algebra generated by the image of $\Phi$: it is the subalgebra of $M_N(\Cbb)$ spanned by the identity 
operator $\mathds 1_N$ and products of $\Phi(X),$
  
$\Phi(X)^*$ for $X\in M_n(\Cbb)$. Given $U\in\mathcal U_n$, we define on $\Acal(\Phi)$ an endomorphism as follows.

We fix a basis of $\Acal(\Phi)$ as a vector space, 
$(\Phi(X_{j_{1}})\cdots\Phi(X_{j_l}))$ where the $j_i$s run in some finite set $J$. The endomorphism $L_U$
is defined by extending by linearity the formula
 \[L_U\left(\Phi(X_{j_{1}})\cdots\Phi(X_{j_{l}})\right)=\Phi(UX_{j_{1}}U^*)\cdots\Phi(UX_{j_{l}}U^*)\,.\]
 It can be noted that thanks to the unitarity assumption, this equation becomes
  \[L_U\left(\Phi(X_{j_{1}})\cdots\Phi(X_{j_{l}})\right)=V\Phi(X_{j_{1}})\cdots\Phi(X_{j_{l}})V^*\,,\]
  therefore for any $Y\in \Acal(\Phi)$, the map is actually 
 \[L_U\,:\,Y\in\Acal(\Phi)\mapsto V(U)\,Y\,V(U)^*\,,\]
 therefore this is actually a $*$-automorphism of $\Acal(\Phi)$.
 
 Let us record two other important properties of 
 $L:\mathcal U_n\to \operatorname{Aut}\left(\Acal(\Phi)\right)\,,\,U\mapsto L_U$. First of all, $L$ is continuous, as any linear map between finite 
 dimensional spaces is continuous and so is the multiplication of operators. Secondly, $L$ is actually a group morphism. Indeed, for any 
 $U_1,U_2\in\mathcal U_n$ and any $X_1,...,X_k\in M_n(\Cbb)$,
 \begin{align*}
  L_{U_1U_2}\left(\Phi(X_1)\cdots\Phi(X_k)\right) 
  & = V(U_1)\,\Phi(U_2X_1U_2^*)\cdots\Phi(U_2X_kU_2^*)\,V(U_1)^*  \\
  & = L_{U_1}\left(\,\Phi(U_2X_1U_2^*)\cdots\Phi(U_2X_kU_2^*)\,\right) \\
  & = L_{U_1}\circ L_{U_2}\left(\Phi(X_1)\cdots\Phi(X_k)\right)\,.
 \end{align*}
So, $L:\mathcal U_n\to \operatorname{Aut}\left(\Acal(\Phi)\right)\,,\,U\mapsto L_U$ is a continuous group morphism.
 
It follows from the classification of
automorphisms of a finite-dimensional $C^*$-algebra \cite[Proposition 2.2.6]{goodman2012coxeter} that there exists $W=W(U)$ in $M_{N}(\mathbb{C})$
that normalizes $\Acal(\Phi)$ such that 
$L_U(X)=WXW^*$ and this $W$ is unique up to a central element in the algebra generated by $\Acal(\Phi)$ and the 
normalizer of the commutant of $\Acal(\Phi)$ (which we will call $\operatorname{NC} \Acal(\Phi)$ for the purpose of this proof). 
 
So, the map
$L:\mathcal U_n\to \operatorname{Aut}\left(\Acal(\Phi)\right)\,,\,U\mapsto L_U$
gives rise to a continous application
$\mathcal U_n\to\mathcal{U}_N/\mathcal{Z}\operatorname{\operatorname{NC}} \Acal(\Phi)\,,\,U\mapsto W$
which implements $L_U$
(the continuity follows from the fact that a continuous bijection between two compact sets is bi-continuous).

In order to complete the proof, we have to prove that this map can be lifted continuously to a map
$\mathcal U_n\to\mathcal{U}_N\,,\,U\mapsto W$ that generates the same $L_U$.
This follows from the fact that, by construction, $W$ can be taken as an element of 
$\operatorname{NC} \Acal(\Phi)$, so, given a minimal projection $P$ in $\mathcal{Z}\operatorname{NC} \Acal(\Phi)$, it is enough to 
show that there exists a continuous choice for $PW$. 

There exists $\theta(U_{1},U_{2}) \in \mathbb{R}$ such that 
\begin{align}\label{eq28}
PW(U_{1})W(U_{2}) & = e^{\iota \theta(U_{1},U_{2}) } PW(U_{1}U_{2})    
\end{align}

Note that it follows from the above that the phase
$e^{\iota \theta(U_{1},U_{2}) }$ is a continuous function of $U_1,U_2$.
Let us prove that $PW$ can undergo a phase modification 
that will turn it into a special unitary representation.
First of all, without loss of generality, we may assume that $\theta(I_n, I_n)=0$, at the possible expense of replacing
$PW$ by $e^{\iota \theta(I_n, I_n) }PW$. 
Then, by continuity looking in a vicinity of the identity, taking the determinant, we get that
$$\det PW(U_1)\det PW(U_2)=e^{N\,\iota \theta(U_{1},U_{2}) }\det PW(U_1U_2)\,,$$
in other words
$$e^{\iota \theta(U_{1},U_{2}) }= (\det PW(U_1)\det PW(U_2) /\det PW(U_1U_2))^{1/N}\,,$$
where the $N$-th root is obtained with the prinicipal branch of the logarithm, which is well defined in a neighbourhood of $0$.

So, we may rewrite Equation \eqref{eq28} as
$$PW(U_{1})PW(U_{2})  = (\det PW(U_1)\det PW(U_2) /\det PW(U_1U_2))^{1/N} PW(U_{1}U_{2}),$$
or equivalently,
$$PW(U_{1})\det PW(U_1)^{-1/N}PW(U_{2})\det PW(U_2)^{-1/N}  = PW(U_{1}U_{2})
(\det PW(U_1U_2))^{-1/N}$$
Therefore, setting
$$P\tilde{W}(U_1)=PW(U_{1})\det W(U_1)^{-1/N},$$
and doing the same for each minimal projection  $P$ in $\mathcal{Z}\operatorname{NC} \Acal(\Phi)$
we see that, in a neighbourhood of $0$,
$$\tilde{W}(U_1)\tilde{W}(U_2)=\tilde{W}(U_1U_2).$$
It is known that a finite dimensional representation of a Lie group is in one to one correspondance with that of its Lie algebra, and therefore that it suffices to define it on a neighbourhood of identity, in the simply connected case.
So, this proves that $\tilde{W}$ extends on $S\mathcal{U}_{n}$ to a representation of $S\mathcal{U}_{n}$
(note that the above is a modification of Bargmann's results \cite{Bargmann:1954gh, article8}).  
We can therefore decompose $\pi$ into irreducible representations: 
\[\mathbb{C}^{N}=\bigoplus_{\lambda}\, E_\lambda\otimes \mathbb{C}^{N_\lambda}\,,\]
where $E_\lambda$ are irreducible representations of $S\Ucal_n$ appearing in the unitary representation $\pi$ with multiplicity $N_\lambda$. We denote by $\pi_\lambda$ the unitary representation of $S\Ucal_n$ on $E_\lambda$ induced by $\pi$.

By \cite[Proposition 22.2]{article11}, corresponding to every irreducible representation $\pi_{\lambda}$ of $S\Ucal_n$, there exists an irreducible representation $\tilde{\pi}_{\lambda}$ of $\Ucal_n$ such that $\tilde{\pi}_{\lambda}|_{S\Ucal_n}=\pi_{\lambda}$. Hence, $\tilde{\pi}=\bigoplus_{\lambda}N_{\lambda}\tilde{\pi}_{\lambda}$ is a unitary representation of $\Ucal_n$.
This complete the proof. 
\end{proof}

Let us note that the last part of the proof can be slightly simplified in case $\Acal(\Phi)= M_{N}(\mathbb{C})$.
Indeed, in this case it is enough to replace $\mathcal{U}_N/\mathcal{Z}\operatorname{NC} \Acal(\Phi)\,,\,U\mapsto W$
by $\mathcal{PU}_N$ (the projective unitary group)
and we can refer to Bargmann's theorem.
It is also interesting to note that no continuity is needed in the assumption, it is granted automatically 
by the theorem 
(together with rationality and analyticity) when one replaces $V$ by $W$.

The next corollary builds a bridge between unitarily equivariant maps and $(a,b)$-unitarily equivariant maps. 
It justifies the computation of the latter in the next section.

\begin{corollary}\label{coro_main}
	 Let $\Phi\in \Bcal(M_n(\Cbb),M_N(\Cbb))$ be a unitarily equivariant map. Then there exist 
	 a finite sequence of pair $(a_i,b_i)$ of natural integers, a partial isometric 
	 map $W:\Cbb^{N}\to\oplus_i (\Cbb^n)^{\otimes a_i}\otimes (\Cbb^n)^{\otimes b_i}$ and \\
	 an  $\oplus_i(a_{i},b_{i})$-unitarily equivariant map $\Psi$ on $M_n(\Cbb)$ such that for all $X\in M_n(\Cbb)$,
	\begin{equation}\label{eq_coro_main}
	W\Phi(X)W^*=\Psi(X)\,.
	\end{equation}
\end{corollary}

\begin{proof}
This is a direct consequence of the representation theory of the unitary group. Indeed, for any representation, there exist a finite sequence of pairs $(a_i,b_i)$ of natural integers such that this representation appears in the unitary representation~(see \cite{sepanski2007compact} for instance)
\[U\mapsto \oplus_i \overline{U}^{\otimes a_i}\otimes U^{\otimes b_i}\,.\]
It means that there exists a partial isometry $W$ from $\Cbb^N$ to $\oplus_i (\Cbb^n)^{\otimes a_i}\otimes (\Cbb^n)^{\otimes b_i}$ such that for all $U\in \Ucal_n$,
\[W\,V(U)W^*=P \, \oplus_i (\overline{U}^{\otimes a_i}\otimes U^{\otimes b_i})\,P\,,\]
where $P$ is the orthogonal projection on some supspace of $\oplus_i((\Cbb^n)^{\otimes a_i}\otimes (\Cbb^n)^{\otimes b_i})$ 
invariant by $\oplus_i (\overline{U}^{\otimes a_i}\otimes U^{\otimes b_i})$ for all $U\in\Ucal_n$. We can then define the map 
$\Psi : X\mapsto W\,\Phi(X)\,W^*$ and it can be readily checked that it is an $\oplus_i(a_{i},b_{i})$-unitarily equivariant map.
\end{proof}

\subsection{Characterization of $(a,b)$-unitarily equivariant maps}\label{sectSchur}

We now deal with the question of characterizing $(a,b)$-unitarily equivariant maps. More explicitly, 
let $a,b \in \mathbb{N}$ and $n \geq 2$. Then, what are all the linear maps 
$$\Phi: M_{n}(\mathbb{C}) \rightarrow M_{n}(\mathbb{C})^{\otimes a} \otimes M_{n}(\mathbb{C})^{\otimes b}$$
such that for every unitary $U \in M_{n}(\mathbb{C})$, $\Phi$ satisfy Equation \eqref{equivariant}. By Lemma \ref{vspace}(1) the set of all 
$(a,b)$-unitarily equivariant maps is a vector space. Characterizing this vector space is then equivalent 
to exhibiting one of its basis. In order to do that, we need some basic results from group representation theory and more precisely the 
Schur-Weyl duality Theorem \cite[Theorem 8.2.10]{ceccherini-silberstein_scarabotti_tolli_2010}. More detail can be found 
in \cite[Chapters 7,8]{ceccherini-silberstein_scarabotti_tolli_2010}.\\

Let us define two unitary representations $\sigma_{k}$ and $\rho_{k}$ on $(\mathbb{C}^{n})^{\otimes k}$, with $a+b=k$, of the symmetric 
group $S_{k}$ and the unitary group $\Ucal_n$ respectively. For all $v_1,...,v_k\in\Cbb^n$, $\pi\in S_k$ and $U\in\Ucal_n$, they are defined as
	\begin{align*}
        & \sigma_k(\pi)\,(v_1\otimes\cdots\otimes v_k)=(v_{\pi^{-1}(1)}\otimes \cdots \otimes v_{\pi^{-1}(k)})\,, \\
        & \rho_k(U)\,(v_1\otimes\cdots\otimes v_k)=(Uv_1\otimes\cdots\otimes Uv_k)\,.
	\end{align*}
We denote by $\sigma_k(\mathbb{C} [S_k])$ (resp. $\rho_k(\mathbb{C}[\Ucal_n])$)
 
the $*$-algebra generated by the representation $\sigma_k$ (resp. $\rho_k$). That is,
\begin{align*}
 & \sigma_k(\mathbb{C}[S_k])=\{\sigma_k(\pi)\,;\,\pi\in S_k\}''=\left\{\sum_{\pi\in S_k}\,f(\pi)\sigma_k(\pi)\quad;\quad f:S_k\to\Cbb\right\}\,,\\
 & \rho_k(\mathbb{C}[\Ucal_n])=\{\rho_k(U)\,;\,U\in \Ucal_n\}''\,,
\end{align*}
where the prime symbol denotes the commutant of a set (we will not need the explicit formula for the elements of $\rho_k(\Ucal_n)$). Then 
Schur-Weyl duality Theorem asserts that $\sigma_k(\mathbb{C}[S_k])$ and  $\rho_k(\mathbb{C}[\Ucal_n])$) are the commutant of each other.
 
In other words, for any operator $C\in M_n(\Cbb)^{\otimes k}$,
\begin{equation}\label{eq_schur_weyl}
 [U^{\otimes k},C]=0\quad \forall \quad U\in\Ucal_n\quad\text{iff}\quad C\in\sigma_k(S_k)\,.
\end{equation}
We refer to \cite[Theorem 8.2.8]{ceccherini-silberstein_scarabotti_tolli_2010} for instance for a presentation of this Theorem. With this 
result in hand, we can now prove one of the main results of this paper.	
	
\begin{theorem}\label{Ueq}
	Let $a,b \in \NN$ and $\Phi: M_{n}(\mathbb{C}) \rightarrow M_{n}(\mathbb{C})^{\otimes a} \otimes M_{n}(\mathbb{C})^{\otimes b}$. 
	Then the following assertions are equivalent.
	
	\begin{enumerate}[label=(\roman*)]\label{a,bcase} 
	\item $\Phi$ is $(a,b)$-unitarily equivariant. 
		\item $[C_{\Phi} , \overline{U}^{\otimes a+1} \otimes U^{\otimes b}]=0$ for all $U\in\Ucal_n$, where $C_\Phi$ is the Choi matrix of $\Phi$.
		\item There exists $f:S_{k+1}\to\Cbb$ such that
		\begin{equation}\label{eq_theo_Ueq}
		 C_{\Phi}=\sum_{\pi\in S_{k+1}}\,f(\pi)\,\left(\theta_{n}^{\otimes a+1}\otimes i_n^{\otimes b}\right)\left[\sigma_{k+1}(\pi)\right]\,,
		\end{equation}
        where $\theta_{n}$ is the transpose map on $M_n(\Cbb)$.
	\end{enumerate}
	
 \end{theorem}
 
\begin{proof}
The equivalence of $(i)$ and $(ii)$ is clear from Theorem \ref{V(U)equiv}. \\
$(ii) \Rightarrow (iii)$.\\
Assume that $C_{\Phi}$ commutes with $(\overline{U}^{\otimes a+1} \otimes U^{\otimes b})$ for all $U\in\Ucal_n$, that is,
\begin{align}\label{commute_condition}
& (\overline{U}^{\otimes a+1} \otimes U^{\otimes b})C_{\Phi}(\overline{U}^{\otimes a+1} \otimes U^{\otimes b})^{*}  = C_{\Phi}
\end{align}
Remark that the  transpose map $\theta_{n} \in \mathcal{B}(M_{n}(\mathbb{C}))$ is $(1,0)$-unitarily equivariant, so that for all $U\in\Ucal_n$ 
and all $C\in M_n(\Cbb)^{\otimes a+b+1}$,
\begin{align}\label{transpose_eq}
(\theta_{n}^{\otimes a+1}\otimes i_n^{\otimes b})(\overline{U}^{\otimes a+1} \otimes U^{\otimes b}\,C\,(\overline{U}^{\otimes a+1} \otimes U^{\otimes b})^*) & =U^{\otimes k+1}\,[(\theta_{n}^{\otimes a+1}\otimes i_n^{\otimes b})C]\,{U^*}^{\otimes k+1}\,.
\end{align}
where  $k=a+b$.
Applying Equations (\ref{commute_condition}) and (\ref{transpose_eq}), we get:
\[
(\theta_{n}^{\otimes a+1} \otimes i_{n}^{\otimes b})[C_{\Phi}]=U^{\otimes k+1}\,(\theta_{n}^{\otimes a+1} \otimes i_{n}^{\otimes b})[C_{\Phi}]\, {U^{*}}^{\otimes k+1}\,.
\]
Consequently, $(\theta_{n}^{\otimes a+1} \otimes i_{n}^{\otimes b})[C_{\Phi}] \in \rho_k(\Ucal_n)'$. By the Schur-Weyl duality given 
in \Cref{eq_schur_weyl}, we obtain that $(\theta_{n}^{\otimes a+1} \otimes i_{n}^{\otimes b})[C_{\Phi}] \in \sigma_{k+1}(\mathbb{C}[S_{k+1}])$, 
which implies (iii).\\
$(iii) \Rightarrow (ii)$ is straightforward using the converse part of Schur-Weyl duality and the previous computation.
\end{proof}

Clearly by point (iii) in Theorem \ref{Ueq}, the set of $(a,b)$-unitarily equivariant maps is isomorphic as 
a vector space to the space $\sigma_{k+1}(\mathbb{C}[S_{k+1}])$. Therefore we get the straightforward corollary:
\begin{corollary}\label{coro_Uab}
A basis of $(a,b)$-unitarily equivariant maps is given by the maps $\{\Phi_\pi : \pi \in S_{k+1}\}$, where $k=a+b$, with the corresponding Choi matrices:
\begin{equation}\label{eq_coro_Uab}
 C_{\Phi_\pi}=(\theta_{n}^{\otimes a+1}\otimes i_n^{\otimes b})[\sigma_{k+1}(\pi)]\,,
\end{equation}
where the matrix $\sigma_{k+1}(\pi)$ acts on vectors $(v_1\otimes\cdots \otimes v_{k+1})$ as
\begin{equation}\label{eq_coro_Uab2}
 \sigma_{k+1}(\pi)\,(v_1\otimes\cdots\otimes v_{k+1})=v_{\pi^{-1}(1)} \otimes \cdots \otimes v_{\pi^{-1}(k+1)}\,.
\end{equation}
In particular, when $n\geq k+1$, then the dimension of this vector space is $(k+1)!$.
\end{corollary}

We discuss the graphical representation of this basis in the next section.

\section{Graphical representation of the Choi matrix of $(a,b)$-unitarily equivariant maps}\label{sect3}
 
In this section we study the graphical representation of $(a,b)$-unitarily equivariant maps for any 
$a,b \in \mathbb{N}$. This gives a visual and very convenient method to compute them. We illustrate this with the $(1,1)$-unitarily 
equivariant maps from \cite{COLLINS2018398}.\\
 
The following table in Figure \ref{graph_table} represents the graphical representation of operators defined on Hilbert spaces, their tensor product and how the operations of 
transpose, multiplication of operators perform graphically. More detail about the graphical calculus can be found in 
\cite{Wood:2015:TNG:2871422.2871425,Collins2010}.
\begin{figure}[h!]
\begin{center}
\begin{tikzpicture}[scale=0.7, transform shape]
\usetikzlibrary{calc}
\node (rect) at (0,0) [draw,thick,minimum width=8cm,minimum height=12cm] (a) {};
\draw[black] (0,6) -- (0,-6);
\draw[black] (-4,5) -- (4,5);
\draw[black] (-4,3) -- (4,3);
\draw[black] (-4,1) -- (4,1);
\draw[black] (-4,-1) -- (4,-1);
\draw[black] (-4,-3) -- (4,-3);
\draw[black] (-4,-3) -- (4,-3);

\draw[black] (1,4.5) -- (2,4.5);
\draw[black] (2,4.9) -- (2,4);
\draw[black] (2,4.9) -- (3,4.5);
\draw[black] (2,4) -- (3,4.5);

\draw[black] (2,3.9) -- (2,3.1);
\draw[black] (2,3.9) -- (1,3.5);
\draw[black] (2,3.1) -- (1,3.5);
\draw[black] (2,3.5) -- (3,3.5);

\draw[black] (1,2) -- (1.5,2);
\draw[black] (1.5,2.5) -- (1.5,1.5);
\draw[black] (1.5,1.5) -- (2.5,1.5);
\draw[black] (2.5,1.5) -- (2.5,2.5);
\draw[black] (2.5,2.5) -- (1.5,2.5);
\draw[black] (2.5,2) -- (3,2);

\draw[black] (1-0.3,0.5) -- (1.5-0.3,0.5);
\draw[black] (1-0.3,-0.5) -- (1.5-0.3,-0.5);
\draw[black] (1.5-0.3,0.8) -- (3-0.3,0.8);
\draw[black] (1.5-0.3,-0.8) -- (3-0.3,-0.8);
\draw[black] (3-0.3,-0.8) -- (3-0.3,0.8);
\draw[black] (3-0.3,0.5) -- (3.5-0.3,0.5);
\draw[black] (3-0.3,-0.5) -- (3.5-0.3,-0.5);
\draw[black] (1.5-0.3,0.8) -- (1.5-0.3,-0.8);

\draw[black] (1,-2) -- (1.5,-2);
\draw[black] (2,-2) -- (2.5,-2);
\draw[black] (3,-2) -- (3.5,-2);
\draw[black] (1.5,-1.75) -- (1.5,-2.25);
\draw[black] (1.5,-1.75) -- (2,-1.75);
\draw[black] (2,-1.75) -- (2,-2.25);
\draw[black] (2,-2.25) -- (1.5,-2.25);

\draw[black] (2.5,-1.75) -- (3,-1.75);
\draw[black] (2.5,-2.25) -- (3,-2.25);
\draw[black] (2.5,-1.75) -- (2.5,-2.25);
\draw[black] (3,-1.75) -- (3,-2.25);

\draw[black] (1.5,-4) -- (2.5,-4);
\draw[black] (1.5,-4) -- (1.5,-5);
\draw[black] (2.5,-5) -- (1.5,-5);
\draw[black] (2.5,-5) -- (2.5,-4);
\draw[black] (1.5, -4.5) .. controls (1,-4.5) and (0.5,-3).. (3,-3.5);
\draw[black] (2.5, -4.5) .. controls (3.5,-4.5) and (3.5,-6).. (1,-5.5);

\draw (2,5.5)node  {Graphical representation};
\draw (-2,5.5)node  {Vectors/Operators};
 
\draw (-2,4.5)node  {{$|v \rangle \in \mathbb{C}^{n}$}};
\draw (2.3,4.5)node  {{$|v \rangle $}};
\draw (-2,3.5)node  {{$\langle v| \in (\mathbb{C}^{n})^{*}$}};
\draw (1.7,3.5)node  {{$\langle v| $}};
\draw (-2,2)node  {{$T \in \mathcal{B}(\mathcal{H})$}};
\draw (2,2)node  {{$T$}};
\draw (-2,0)node  {{$T \in \mathcal{B}(\mathcal{H}_{1} \otimes \mathcal{H}_{2})$}};
\draw (2,0)node  {{$T$}};
\draw (0.7-0.3,0.5)node  {{$\mathcal{H}_{1}$}};
\draw (0.7-0.3,-0.5)node  {{$\mathcal{H}_{2}$}};
\draw (3.7-0.3,0.5)node  {{$\mathcal{H}_{1}$}};
\draw (3.7-0.3,-0.5)node  {{$\mathcal{H}_{2}$}};
\draw (-2,-2)node  {{$AB \in M_{n}(\mathbb{C})$}};
\draw (1.75,-2)node  {{$A$}};
\draw (2.75,-2)node  {{$B$}};
\draw (-2,-4.5)node  {{$A^{t}\in M_{n}(\mathbb{C})$}};
\draw (2,-4.5)node  {{$A$}};
\end{tikzpicture}
\end{center}
\caption{Graphical representations of vectors and operators}
\label{graph_table}
\end{figure}
\newpage 
 Let $W_{a,b}$ denote the vector space of all $(a,b)$-unitarily equivariant maps. Recall from \Cref{sectSchur} 
 that the maps $C_{\Phi_\pi}=(\theta_n^{\otimes a+1}\otimes i_n^{\otimes b})(\sigma_{k+1}(\pi))$ form a basis of $W_{a,b}$ with $\pi$ running 
 through the permutation group $S_{k+1}$, where $\sigma_{k+1}$ is the unitary representation of the permutation group defined in 
 \Cref{eq_coro_Uab2}. Based on the graphical representations given in Table \ref{graph_table}, we give the graphical representations of the 
 $C_{\Phi_\pi}$. We first consider the case of $a=1$ and $b=1$ and make some observations. We discuss the general case in Theorem 
 \ref{graphical_mainthm}.
 
We compute the graphical representation of the Choi matrices $C_{\Phi_{\pi}}$ with $\pi\in S_3$, by first computing the one of $\sigma_{3}(\pi)$. We start with the representation of $\sigma_{3}(123)$ as an example. This map is defined as follows:
\begin{equation*} 
\begin{array}{cccc}
\sigma_{3}(123) :& (\mathbb{C}^{n})^{\otimes 3} &\rightarrow &  (\mathbb{C}^{n})^{\otimes 3}  \\
& v_{1} \otimes v_{2} \otimes v_{3}               & \mapsto    & v_{3} \otimes v_{1} \otimes v_{2}
\end{array}\,.
\end{equation*}
Its graphical representation is given on the left in \Cref{graph_sigma3(123)}. By \Cref{eq_coro_Uab}, 
$C_{\Phi_{(123)}}= (\theta_{n} ^{\otimes 2} \otimes i_n)\sigma_{3}(123)$. Using the graphical representation of the transpose map, 
we obtain that $C_{\Phi_{(123)}}$ is represented by the right figure in \Cref{graph_sigma3(123)}.\\
 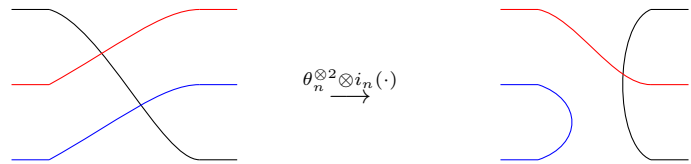
\begin{figure}[h!]
 \begin{center}
 \begin{tikzpicture}	
 	\usetikzlibrary{calc}
\draw[black] (-4,1) -- (-3.5,1);
\draw[red] (-4,0) -- (-3.5,0);
\draw[blue] (-4,-1) -- (-3.5,-1);
\draw[red] (-1.5,1) -- (-1,1);
\draw[blue] (-1.5,0) -- (-1,0);
\draw[black] (-1.5,-1) -- (-1,-1);

\draw[black] (-3.5, 1) .. controls (-2.8,0.8) and (-2,-1).. (-1.5,-1);
\draw[red] (-3.5, 0) .. controls (-2.6,0.5) and (-2,1).. (-1.5,1);
\draw[blue] (-3.5, -1) .. controls (-2.6,-0.5) and (-2,0).. (-1.5,0);

\draw (0.5,0)node  {{$\overset{\theta_n^{\otimes 2}\otimes i_n(\cdot)}\longrightarrow$}};  

	\draw[red] (2.5,1) -- (3,1);
	\draw[blue] (2.5,0) -- (3,0);
	\draw[blue] (2.5,-1) -- (3,-1);
	\draw[black] (5,1) -- (4.5,1);
	\draw[red] (5,0) -- (4.5,0);
	\draw[black] (5,-1) -- (4.5,-1);
	\draw[blue] (3, 0) .. controls (3.6,-0.2) and (3.6,-0.8).. (3,-1);
	\draw[black] (4.5, 1) .. controls (4,0.8) and (4,-0.8).. (4.5,-1);
	\draw[red] (3, 1) .. controls (3.5,0.9) and (4,0).. (4.5,0);
	
\end{tikzpicture}
\end{center}
\caption{Graphical representation of $\sigma_3(123)$ (on the left) and $C_{\Phi_{(123)}}$ (on the right)}
\label{graph_sigma3(123)}	
\end{figure}

We take the following convention in order to represent the operators $\sigma_{k+1}(\pi)$: input on the left are represented by black (full) dots 
while output on the right are represented by white (empty) dots. The graphical representation of the operator $\sigma_{k+1}(\pi)$ then 
corresponds to tracing a wire from the $i$th (black) dot on the left to the $\pi^{-1}(i)$ (white) dot on the right. Using graphical computation, it can 
be directly check that $\sigma_{k+1}(\pi)$ and $U^{\otimes k+1}$ commute for all unitary operator $U\in\Ucal_n$.\\
Then, applying the transpose map to the first $(a+1)$ tensors corresponds graphically to interchange the black and the white dot at the first $(a+1)$ 
rows (starting from the highest). Again, it can be directly checked from a graphical computation that $C_{\Phi(\pi)}$ and 
$(\overline U^{\otimes a+1}\otimes U^{\otimes b})$ commute. This is illustrated in the case of the permutation $(123)$ in \Cref{graph_final}.
 
\begin{figure}
\begin{center}
\begin{tikzpicture}
\node (rect) at (-4,0) [draw,thick,minimum width=2cm,minimum height=3cm] (a) {};
\node (rect) at (-2-4,1.2) [draw,thick,minimum width=0.5cm,minimum height=0.5cm] {$\overline{U}$};
\node (rect) at (-2-4,0) [draw,thick,minimum width=0.5cm,minimum height=0.5cm] {$\overline{U}$};
\node (rect) at (-2-4,-1.2) [draw,thick,minimum width=0.5cm,minimum height=0.5cm] {$U$};
\node (rect) at (2-4,1.2) [draw,thick,minimum width=0.5cm,minimum height=0.5cm] {$U^{t}$};
\node (rect) at (2-4,0) [draw,thick,minimum width=0.5cm,minimum height=0.5cm] {$U^{t}$};
\node (rect) at (2-4,-1.2) [draw,thick,minimum width=0.5cm,minimum height=0.5cm] {$U^{*}$};
\draw[black] (-3-4,1.2) -- (-2.3-4,1.2);
\draw[black] (-1.7-4,1.2) -- (-1-4,1.2);

\draw[black] (-3-4,0) -- (-2.3-4,0);
\draw[black] (-1.7-4,0) -- (-1-4,0);

\draw[black] (-3-4,-1.2) -- (-2.3-4,-1.2);
\draw[black] (-1.7-4,-1.2) -- (-1-4,-1.2);

\draw[black] (2.3-4,1.2) -- (3-4,1.2);
\draw[black] (1-4,1.2) -- (1.7-4,1.2);

\draw[black] (2.3-4,0) -- (3-4,0);
\draw[black] (1-4,0) -- (1.7-4,0);

\draw[black] (2.3-4,-1.2) -- (3-4,-1.2);
\draw[black] (1-4,-1.2) -- (1.7-4,-1.2);

\draw (0,0)node  {{$ =$}}; 
\draw[black] (-1-4, 0) .. controls (-0.4-4,-0.2) and (-0.4-4,-0.8).. (-1-4,-1.2);
\draw[black] (1-4, 1.2) .. controls (0.4-4,0.8) and (0.4-4,-0.8).. (1-4,-1.2);
\draw[black] (-1-4, 1.2) .. controls (-0.5-4,0.9) and (0.5-4,-0.3).. (1-4,0);

\path node at ( 2,0) [shape=circle,fill] {};
\node at ( 2,-1.2) [shape=circle,draw] {}
	 
node at (4,1.2) [shape=circle,draw] {}
node at (4,-1.2) [shape=circle,fill] {}
node at (2,1.2) [shape=circle,fill] {}
node at (4,0) [shape=circle,draw] {};
 
\draw[black] (2, 0) .. controls (2.8,-0.2) and (2.8,-1).. (2.15,-1.2);
\draw[black] (3.85, 1.2) .. controls (3.2,1) and (3.2,-1).. (4,-1.2);
\draw[black] (2, 1.2) .. controls (2.5,0.9) and (3.5,-0.3).. (3.85,0);
\end{tikzpicture}
\end{center}
\caption{Graph of $C_{\Phi_{(123)}}$. Each black dot is connected to a white dot in a one-to-one way, according to the permutation $(123)$.}
\label{graph_final}
\end{figure}
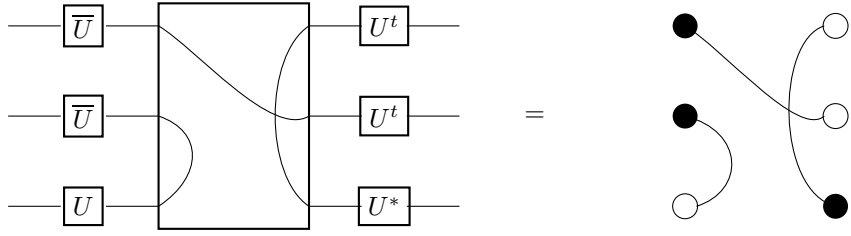

We can obtain similarly the graphical representations of the Choi matrices $C_{\Phi_{\pi}}$ corresponding to all $\pi \in  S_{3}$. There are 
$3!=6$ different possibilities to trace a wire between black and white dots in a one-to-one way. Each of them corresponds to a different 
permutation $\pi\in S_3$. This is illustrated in \Cref{graph_complete}.

\begin{center}
	\begin{figure}[h!]
	
\begin{center}
		\begin{tikzpicture}[scale=0.8]
		\usetikzlibrary{calc}
		\node (rect) at (0,0) [draw,thick,minimum width=11cm,minimum height=3.5cm] (a) {};
		\draw[black] (-6.85,1) -- (6.85,1);
		\draw[black] (-5,2.2) -- (-5,-2.2);
		\draw (-6,1.5)node  {{$\pi$}};
		\draw (-6,-.5)node  {{$C_{\Phi_{\pi}} $}};
		\draw (-4,1.5)node   {{ $\mathds 1$}};
		\draw (-2,1.5)node  {{$ (23)$}};
		\draw (0,1.5)node  {{$ (12) $}};
		\draw (2,1.5)node  {{$ (13)$}};
		\draw (4,1.5)node  {{$ (123)$}};
		\draw (6,1.5)node  {{$(132) $}};
		\path node at (-4.5,0.5) [shape=circle,fill] {};
		\node at (-4.5,-0.5) [shape=circle,fill] {}
		node at (-4.5,-1.5) [shape=circle,draw] {}
		node at (-3.5,0.5) [shape=circle,draw] {}
		node at (-3.5,-0.5) [shape=circle,draw] {}
		node at (-3.5,-1.5) [shape=circle,fill] {}
		
		node at (-2.5,0.5) [shape=circle,fill] {}
		node at (-2.5,-0.5) [shape=circle,fill] {}
		node at (-2.5,-1.5) [shape=circle,draw] {}
		node at (-1.5,0.5) [shape=circle,draw] {}
		node at (-1.5,-0.5) [shape=circle,draw] {}
		node at (-1.5,-1.5) [shape=circle,fill] {}
		
		node at (-.5,.5) [shape=circle,fill] {}
		node at (-.5,-.5) [shape=circle,fill] {}
		node at (-.5,-1.5) [shape=circle,draw] {}
		node at (.5,.5) [shape=circle,draw] {}
		node at (.5,-.5) [shape=circle,draw] {}
		node at (.5,-1.5) [shape=circle,fill] {}
		
		node at (1.5,0.5) [shape=circle,fill] {}
		node at (1.5,-0.5) [shape=circle,fill] {}
		node at (1.5,-1.5) [shape=circle,draw] {}
		node at (2.5,0.5) [shape=circle,draw] {}
		node at (2.5,-0.5) [shape=circle,draw] {}
		node at (2.5,-1.5) [shape=circle,fill] {}
		
		node at (3.5,0.5) [shape=circle,fill] {}
		node at (3.5,-0.5) [shape=circle,fill] {}
		node at (3.5,-1.5) [shape=circle,draw] {}
		node at (4.5,0.5) [shape=circle,draw] {}
		node at (4.5,-0.5) [shape=circle,draw] {}
		node at (4.5,-1.5) [shape=circle,fill] {}
		
		node at (5.5,0.5) [shape=circle,fill] {}
		node at (5.5,-0.5) [shape=circle,fill] {}
		node at (5.5,-1.5) [shape=circle,draw] {}
		node at (6.5,0.5) [shape=circle,draw] {}
		node at (6.5,-0.5) [shape=circle,draw] {}
		node at (6.5,-1.5) [shape=circle,fill] {};

		\draw[black] (-4.5, 0.5) -- (-3.7,0.5);
		\draw[black] (-4.3, -1.5)  -- (-3.7,-1.5);
		\draw[black] (-4.7,-0.5)  --  (-3.7,-0.5);
	 	
		\draw[black] (-2.5, 0.5)  -- (-1.7,.5);
		\draw[black] (-2.5, -0.5) .. controls (-2,-0.4)  and (-2,-1.5)  .. (-2.3,-1.5);
		\draw[black] (-1.7, -.5) .. controls (-2,-.4)  and (-2,-1.5)  .. (-1.7 ,-1.5);
	 			
		\draw[black] (-.5, 0.5) .. controls (-0.2,0.3)  and (0,-.1)  .. (0.3,-.5);
		\draw[black] (-.3, -1.5)  -- (0.3,-1.5);
		\draw[black] (-0.5, -.5) .. controls (-0.5,-1.2)  and (0.2,0)  .. (.3,0.5);
 			
		\draw[black] (1.3,.5) .. controls (2,0.3)  and (2,-1.3)  .. (1.7,-1.5);
		\draw[black] (1.5, -.5)  -- (2.3,-0.5);
		\draw[black] (2.3, .5) .. controls (2,0.3)  and (2,-1.3)  .. (2.5,-1.5);
	 
		\draw[black] (3.5,.5) .. controls (3.7,0.3)  and (4.1,-0.2)  .. (4.3,-.5);
		\draw[black] (4.3, .5) .. controls (4,0.2)  and (4,-1.3)  .. (4.5,-1.5);
		\draw[black] (3.5, -.5) .. controls (4,-0.5)  and (4,-1.5)  .. (3.7,-1.5);
	 					
		\draw[black] (5.5, .5)  .. controls (6,0.3)  and (6,-1.3)  .. (5.7,-1.5);
		\draw[black] (5.5, -.5) .. controls (5.7,0)  and (6.2,.3)  .. (6.3,.5);
		\draw[black] (6.3, -.5) .. controls (6,-0.5)  and (6,-1.6)  .. (6.5,-1.5);
		\draw[black] (-3,2.2) -- (-3,-2.2);
		\draw[black] (-1,2.2) -- (-1,-2.2);
		\draw[black] (1,2.2) -- (1,-2.2);
		\draw[black] (3,2.2) -- (3,-2.2);
		\draw[black] (5,2.2) -- (5,-2.2);
		\end{tikzpicture} 
		
\end{center}
		\caption{Graphical representations of the $C_{\Phi_\pi}$ for each $\pi\in S_3$}
		\label{graph_complete}	
	\end{figure}
\end{center}

Figure \ref{graph_complete} gives the basis elements of the vector space $W_{1,1}$. Theorem \ref{graphical_mainthm} generalizes this fact to 
any $(a,b)$-unitarily equivariant map. We omit the proof, as it follows exactly the same idea as the example 
above.

\begin{theorem}\label{graphical_mainthm}
The graphical representation of the Choi matrices of the basis elements $C_{\Phi(\pi)}$ of the vector space $W_{a,b}$ is given by the following 
rule. Being a matrix on $({\Cbb^n})^{\otimes k+1}$, its graphical representative has $(k+1)$ input (on the left) and $(k+1)$ output (on the right). The 
first $(a+1)$ inputs are symbolized by black (full) dots, the remaining $b$ by white (empty) dots. In the same way, the first $(a+1)$ outputs are 
symbolized by white (empty) dots, the remaining $b$ by black (full) dots. Then each element of this basis corresponds to a possible wiring 
between black and white dot, in a one-to-one way.
\end{theorem}

Theorem \ref{graphical_mainthm} gives a graphical representation of the Choi matrices of the basis elements of $W_{a,b}$, which completes the 
characterization of the vector space $W_{a,b}$ of all $(a,b)$-unitarily equivariant maps. As an illustration, we 
study the graphical representation of the Choi matrix of the family of linear maps $\Phi_{\alpha,\beta,3}$ studied in \cite{COLLINS2018398} and 
given by \Cref{eq_Phi_alphabeta}.
 It is given as follows:\\

\begin{figure}[h]
	\begin{center}
		\begin{tikzpicture}[scale=0.7]
		
		\path node at (-4.5,1) [shape=circle,fill] {};
		\node at (-4.5,0) [shape=circle,fill] {}
		node at (-4.5,-1) [shape=circle,draw] {}
		node at (-3.5,1) [shape=circle,draw] {}
		node at (-3.5,0) [shape=circle,draw] {}
		node at (-3.5,-1) [shape=circle,fill] {}
		
		node at (-1.5,1) [shape=circle,fill] {}
		node at (-1.5,0) [shape=circle,fill] {}
		node at (-1.5,-1) [shape=circle,draw] {}
		node at (-0.5,1) [shape=circle,draw] {}
		node at (-0.5,0) [shape=circle,draw] {}
		node at (-0.5,-1) [shape=circle,fill] {}
		
		node at (1.5,1) [shape=circle,fill] {}
		node at (1.5,0) [shape=circle,fill] {}
		node at (1.5,-1) [shape=circle,draw] {}
		node at (2.5,1) [shape=circle,draw] {}
		node at (2.5,0) [shape=circle,draw] {}
		node at (2.5,-1) [shape=circle,fill] {}
		
		node at (4.5,1) [shape=circle,fill] {}
		node at (4.5,0) [shape=circle,fill] {}
		node at (4.5,-1) [shape=circle,draw] {}
		node at (5.5,1) [shape=circle,draw] {}
		node at (5.5,0) [shape=circle,draw] {}
		node at (5.5,-1) [shape=circle,fill] {};
		\draw (-2.5,0)node  {{$+$}};
		\draw (0.5,0)node  {{$+$}};
		\draw (0.8,0)node  {{$\alpha$}};
		\draw (3.5,0)node  {{$+$}};
		\draw (4,0)node  {{$\beta$}};
 
\draw[black] (-4.5, 1)  -- (-3.7,0);
\draw[black] (-4.5, 0)  -- (-3.7,1);
\draw[black] (-4.3, -1)  -- (-3.7,-1);

\draw[black] (-1.5, 1) .. controls (-1,0.2)  and (-1,-0.8)  .. (-1.3,-1);
\draw[black] (-.7, 1) .. controls (-1,0.2)  and (-1,-0.8)  .. (-.3,-1);	
\draw[black] (-1.5, 0)  -- (-0.7,0);	

\draw[black] (1.5, 0)  -- (2.3,0);
\draw[black] (1.5, 1)  -- (2.3,1);
\draw[black] (1.7, -1)  -- (2.5,-1); 

\draw[black] (4.5, 0) .. controls (5,0.1)  and (5,-1)  .. (4.7,-1);
\draw[black] (5.3, 0) .. controls (5,0.2)  and (5,-1.5)  .. (5.7,-1);	
\draw[black] (4.5, 1)  -- (5.3,1);	
		\end{tikzpicture}
	\end{center}
	\caption{Graphical representation of $C_{\Phi_{\alpha,\beta,n}}$}
	\label{graph_anexample}
\end{figure}
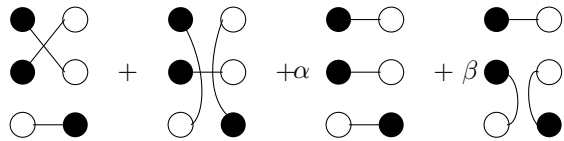
Remark that it corresponds to taking linear combinations of the graphical representations of the permutations $(12)$, $(13)$, $(1)$ and $(23)$ 
respectively and from left to right. A more general $3$-parameters family can be defined by adding the permutations $(123)$ and $(132)$ 
and it gives:
\begin{equation}
\begin{array}{cccc}
	\Phi_{\alpha, \beta,\gamma,n} :& M_{n}(\mathbb{C}) &\rightarrow & M_{n}(\mathbb C) \otimes M_{n}(\mathbb C)  \\
	                          & A               & \mapsto &  A^{t} \otimes \mathds{1}_{n}+ \mathds{1}_{n} \otimes A + \mathrm{Tr}(A)(\alpha \mathds{1}_{n^{2}}+\beta B_{n^{2}})\\
	                          & & & +\gamma\left(B_{n^2}(\mathds{1}_{n} \otimes A)+(\mathds{1}_{n} \otimes A) B_{n^2}\right)\,. 
	                          \end{array}
\end{equation}
We leave the study of this more general map to future work.

\section{Application to entanglement detection}\label{sect4}

One application of positive maps that are not completely positive is entanglement detection in quantum information theory. We recall that a 
density matrix in $M_n(\Cbb)$ for some integer $n\geq0$ is a positive semi-definite operator on $\Cbb^n$ with trace equal to one. We denote by 
$\Scal(n)$ the set of density matrices on $\Cbb^n$ and by $\Scal(m\times n)$ the set of density matrices on $(\Cbb^m\otimes\Cbb^n)$. Then a 
density matrix on $(\Cbb^m\otimes\Cbb^n)$ is called \emph{separable} if it is in the convex hull of tensor-product density matrices. We denote by 
$\SEP(m,n)$ the set of separable density matrices on $(\Cbb^m\otimes\Cbb^n)$, that is,
\[\SEP(m,n)=\left\{\sum_i\,\lambda_i\rho_A^i\otimes\rho_B^i\,;\,\rho_A^i\in\Scal(m),\rho_B^i\in\Scal(n)\,,\,\lambda_i\geq0,\sum_i\lambda_i=1\right\}\,.\]
We are interested in  the complement of this set in $\Scal(m\times n)$.
If a density matrix 
$\rho$ is called $\Scal(m\times n)\backslash \SEP(n,m)$, it is called \emph{entangled}. The first operational characterization of 
entanglement/separability was proved in \cite{HORODECKI19961}: a density matrix $\rho\in\Scal(m\times n)$ is entangled if and only 
if there exists a positive map $\phi\,:\,M_n(\Cbb)\to M_m(\Cbb)$ such that $(i_m\otimes\phi)(\rho)$ is not positive semi-definite. 
Remark that if $\phi$ is completely positive, then $(i_m\otimes\phi)(\rho)$ is necessarily positive semi-definite. This means that 
positive -- but not completely positive - maps lead to \emph{entanglement detectors}. The most well-known example of positive maps 
leading to entanglement detection is the partial transpose, and the density matrices that remain positive under its action are the 
so-called \emph{PPT density matrices}. More generally, given a positive map $\phi\,:\,M_n(\Cbb)\to M_N(\Cbb)$ for any natural number $N>0$, we can define
\begin{equation}\label{eq_entanglementwitness}
 \Scal_m(\phi)=\{\rho\in\Scal(m\times n)\,;\,(i_m\otimes\phi)(\rho)\geq0\}\,.
\end{equation}
Thus the set of PPT density matrices on $\Cbb^m\otimes\Cbb^n$ is $\Scal_m(\theta_{n})$, where $\theta_{n}$ is the transposition on 
$M_n(\Cbb)$. Apart from $(m,n)=(2,2)$ or $(m,n)=(2,3)$, it is known that $\Scal_m(\theta_{n})$ is strictly different from $\SEP(m,n)$. 
In fact, deciding whether a density matrix is separable or not is known as a hard computational problem 
\cite{Gharibian:2010:SNQ:2011350.2011361,gurvits2003classical}. 
It means in particular that an infinite 
number of positive maps are needed to detect all entangled density matrices.  
It is thus a central question in quantum information theory 
to find families of maps that are sufficient for this task.

More generally, it is also interesting to detect the amount of entanglement contains in a bipartite density matrix. 
The \emph{Schmidt number} represents in this respect a good integer quantifier. We recall its definition. Recall that any norm-one 
vector $\psi\in (\Cbb^m\otimes\Cbb^n)$ admits a unique Schmidt decomposition as
\[\psi=\sum_{i=1}^t\,\sqrt{\lambda_i}\,\psi^m_i\otimes\psi^n_i\,,\]
where $(\psi_i^m)_{i=1,...,t}$ (resp. $(\psi_i^n)_{i=1,...,t}$) is an orthogonal family on $\Cbb^m$ (resp. on $\Cbb^n$) with consequently 
$1\leq t\leq\min\{m,n\}$, and where $\sum_{i=1}^t\,\lambda_i=1$. Then its \emph{Schmidt rank} is defined as $\operatorname{SR}(\psi):=t$. 
The Schmidt number of a density matrix $\rho\in\Scal(m\times n)$ is defined subsequently as
\[\operatorname{SN}(\rho):=\min\left\{{\underset{j} \max}\,\operatorname{SR}(\psi_j)\,;\,\rho=\sum_j\,\rho_j|\psi_j\rangle\langle\psi_j|\right\}\,,\]
where the minimum is over all possible decompositions of $\rho$ as a linear combination of pure states. 

We denote by $\operatorname{SN}_t(m,n)$ the set of density matrices with Schmidt number smaller than $t\geq1$. 
We shall also refer 
to elements of this set as \emph{$t$-separable} density matrices and elements of its complement 
$\Scal(m\times n)\backslash\operatorname{SN}_t(m,n)$ as \emph{$t$-entangled} density matrices. Remark in particular that the 
separable density matrices are exactly the ones with Schmidt number equal to one: $\SEP(m,n)=\operatorname{SN}_1(m,n)$. 
In the same way that separability (or $1$-separability) can be checked using positive maps, $t$-separability can be checked using 
$t$-positive maps. That is, a density matrix $\rho\in\Scal(m\times n)$ is $t$-separable if and only if $(i_m\otimes \phi)(\rho)\geq0$ for 
all $t$-positive maps $\phi:M_n(\Cbb)\to M_n(\Cbb)$. \\

We are now in position to state the main result of this article.
\begin{theorem}\label{theo_entanglementwitness}
Let $\phi\,:\,M_n(\Cbb)\to M_N(\Cbb)$ be a $t$-positive - but not completely positive - map, with $1\leq t\leq \min\{m,n\}$. Then there 
exists a family of 
unitarily equivariant $t$-positive maps $(\Phi_a)_{a\geq1}$ with 
$\Phi_a\,:\,M_n(\Cbb)\to M_{N_a}(\Cbb)$, $\underset{a\to+\infty}{\lim} {N_a}=+\infty$, 
such that $(\Scal_m(\Phi_a))_{a\geq1}$ form a decreasing family of sets and:
\begin{equation}\label{eq_theo_entanglementwitness}
  \underset{a\geq1}{\cap}\Scal_m(\Phi_a)\subset\Scal_m(\phi)\,.
\end{equation}
In other words, for all $t$-entangled density matrices $\rho\notin\Scal_m(\phi)$, there exists an integer $a_0\geq1$ such that 
$\rho\notin\Scal_m(\Phi_a)$ for all $a\geq a_0$.
\end{theorem}

We divide the proof of \Cref{theo_entanglementwitness} in two parts: first in \Cref{sect41}, we prove the existence of a 
$\Phi:M_n(\Cbb)\to \Bcal(\Hcal)$, for some \emph{universal infinite dimensional} Hilbert space $\Hcal$ and universal unitary representation 
$U\mapsto V(U)$ that fulfills the equivariance property \eqref{eqgeneral}, such that $\Scal_m(\Phi)\subset\Scal_m(\phi)$. Then we prove the 
theorem by an application of the Peter-Weyl Theorem together with \Cref{coro_main}.

\subsection{Proof of Theorem \ref{theo_entanglementwitness} with infinite dimensional image}\label{sect41}

\begin{theorem}\label{theo_entanglementwitness_infinite}
Let $\phi\,:\,M_n(\Cbb)\to M_N(\Cbb)$ be a $t$-positive map. Then there exists a unitarily equivariant $t$-positive map 
$\Phi\,:\,M_n(\Cbb)\to \Bcal(\Hcal)$ for some Hilbert space $\Hcal$ such that:
\begin{equation}\label{eq_theo_entanglementwitness}
  \Scal_m(\Phi)\subset\Scal_m(\phi)\,,
\end{equation}
where these quantities were defined in Equation \eqref{eq_entanglementwitness}.
\end{theorem}

Before proving the theorem, we recall some facts about the left regular representation of the unitary group $\Ucal_n$. Recall that 
$L_2(\Ucal_n)$ is the Hilbert space of square integrable functions on $\Ucal_n$ with respect to the Haar measure:
\[L_2(\Ucal_n)=\{f:\Ucal_n\to\Cbb\,;\,\int_{\Ucal_n}\,|f(U)|^2\mu_{\operatorname{Haar}}(dU)<+\infty\}\,.\]
The \emph{left regular representation} $\lambda$ of $\Ucal_n$ on $L_2(\Ucal_n)$ is defined for all $W\in\Ucal_n$ as:
\begin{equation}\label{eq_leftregular}
\lambda_W\,f\in L_2(\Ucal_n)\mapsto \left(U\mapsto f(W^*U)\right)\,. 
\end{equation}

Now define $\Hcal=L_2(\Ucal_n,\Cbb^N)\approx L_2(\Ucal_n)\otimes\Cbb^N$, that is, 
 \begin{equation}\label{eq_H}
\Hcal=\left\{(\psi_U)_{U\in \Ucal_n}\,;\,\psi_U\in\Cbb^N\ \forall U\in\Ucal_n\,,\quad\int_{\Ucal_n}\,\|\psi_U\|^2\mu_{\operatorname{Haar}}(dU)<+\infty\right\}  
 \end{equation}
where the integration is with respect to the Haar measure on $\Ucal_n$. The Hilbert space $\Hcal$ will play the role of the image of 
$\Phi$ in \Cref{theo_entanglementwitness_infinite}. The equivariance property will be given in terms of the following unitary 
representation of $\Ucal_n$ on $\Hcal$:
 \begin{equation}\label{eq_rep}
V:W\mapsto \big(V(W):\psi\in\Hcal\mapsto(\psi_{W^*U})_{U\in\Ucal_n}\big)\,.  
 \end{equation}
 
We can readily check that $V$ is the left regular representation on $L_2(\Ucal_n)$ tensored with the trivial representation on 
$\Cbb^N$, that is, $V=\lambda\otimes i_N$.

Finally, let $\phi\in\Bcal(M_n(\Cbb),M_N(\Cbb))$ be a linear map. We can define a map $\Phi\in\Bcal(M_n(\Cbb),\Bcal(\Hcal))$ as
\begin{equation}\label{eq_Phi}
X\in M_n(\Cbb)\mapsto\Phi(X)=\left(\phi(U^*XU)\right)_{U\in\Ucal_n}\,. 
\end{equation}
Thus for all $X\in M_n(\Cbb)$ and all $(\psi_U)_{U\in \Ucal_n}\in\Hcal$, $\Phi(X)\psi=(\phi(U^*XU)\psi_U)_{U\in \Ucal_n}$.

\begin{lemma}\label{lema_Phi}
Let $\phi\in\Bcal(M_n(\Cbb),M_N(\Cbb))$ be a linear map and define the linear map $\Phi:M_n(\Cbb)\to\Bcal(\Hcal)$ as in \Cref{eq_Phi}. 
Then $\Phi$ is unitarily equivariant with respect to the unitary representation $U\mapsto V(U)$. If furthermore $\phi$ is $t$-positive, 
then $\Phi$ is also $t$-positive.
\end{lemma}

\begin{proof}
 We first prove that $\Phi$ is $U\mapsto V(U)$ equivariant. Indeed, for all $W\in\Ucal_n$, $X\in M_n(\Cbb)$ and all $\psi\in\Hcal$,
 \begin{align*}
  V(W)^*\,\Phi(WXW^*)\,V(W)\,\psi 
  & = V(W)^*\,\Phi(WXW^*)\,(\psi_{W^{*}U})_{U\in\Ucal_n} \\
  & = V(W^*)\,\big(\phi((W^{*}U)^*\,X\,(W^{*}U))\psi_{W^{*}U}\big)_{U\in\Ucal_n} \\
  & = \Phi(X)\psi\,.
\end{align*}
Secondly, we prove that $\Phi$ is $t$-positive. As it is equivariant, by \Cref{k-positivity}, we only need to check that the matrix 
$(\Phi(e_{ij}))_{1\leq i,j\leq t}$ is positive semi-definite, where $(e_{ij})_{1\leq i,j\leq n}$ are the matrix units in $M_n(\Cbb)$. This can be directly 
checked by noticing that $\Hcal\otimes\Cbb^t\simeq L_2(\Ucal_n,\Cbb^{tN})$ so that
\[(\Phi(e_{ij}))_{1\leq i,j\leq t}\simeq\left((\phi(U^*\,e_{ij}\,U))_{1\leq i,j\leq t}\right)_{U\in \Ucal_n}\,.\]
If $\phi$ is $t$-positive, then $(\phi(U^*\,e_{ij}\,U))_{1\leq i,j\leq t}$ is positive semi-definite for all $U\in\Ucal_n$, and so is 
$(\Phi(e_{ij}))_{1\leq i,j\leq t}$.
\end{proof}
 
\begin{proof}[Proof of \Cref{theo_entanglementwitness_infinite}]
Let $\rho\in\Scal(m\times n)$ be a $t$-entangled density matrix such that $(i_m\otimes\phi)(\rho)$ is not positive semi-definite, i.e. 
$\rho\notin\Scal_m(\phi)$. We can now show that $\Phi$ ``detects'' $\rho$, that is, $(i_m\otimes\Phi)(\rho)$ is not positive semi-definite. 
Indeed, this is true as by continuity there exist a vector $\varphi\in(\mathbb{C}^m\otimes\mathbb{C}^N)$ and a small neighborhood $\Ucal$ of 
$I_n$ in $\Ucal_n$ such that for all $U\in\Ucal$, $\langle\varphi\,,(i_m\otimes\phi)(I_m\otimes U\,\rho\,I_m\otimes U^*))\varphi\rangle<0$. Then 
defining the vector $\psi\in(\Cbb^m\otimes\Hcal)$ as:
\[\psi_U=\varphi\quad\forall U\in\Ucal\,,\qquad\psi_U=0\quad\text{ elsewhere }\,,\]
we get that $\langle\psi\,,\,(i_m\otimes\Phi)(\rho)\psi\rangle<0$ which shows that $(i_m\otimes\Phi)(\rho)$ is not positive semi-definite.
\end{proof}

\subsection{Proof of \Cref{theo_entanglementwitness}}

Going from \Cref{theo_entanglementwitness_infinite} to \Cref{theo_entanglementwitness} is a simple application of the Peter-Weyl 
Theorem applied to the compact group $\Ucal_n$. It states that $L_2(\Ucal_n)$ is 

isomorphic as a $\Ucal_n$--space to

the closure of the direct sum of all finite dimensional 
irreducible unitary representations $E_\pi$ of $\Ucal_n$, each with multiplicity equal to its dimension (see \cite{article11} for instance):
\begin{equation}\label{eq_theo_PeterWeyl}
 L_2(\Ucal_n)=\underset{\pi\in\Lambda}{\widehat {\bigoplus}}\,E_\pi\otimes F_\pi\,,
\end{equation}
where $\Lambda$ is the \emph{countable} set that indexes all irreps of $\Ucal_n$ and where $F_\pi$ is of dimension $\dim E_\pi$. 
Furthermore, the left-regular representation $\lambda$ of $\Ucal_n$ on $L_2(\Ucal_n)$ given in \Cref{eq_leftregular} can be 
decomposed according to \Cref{eq_theo_PeterWeyl}:
\begin{equation}\label{eq_leftregular_PeterWeyl}
 \lambda=\underset{\pi\in\Lambda}{\widehat {\oplus}}\,\lambda^{(\pi)}\otimes i_{F_\pi}\,,
\end{equation}
where $\lambda^{(\pi)}$ is the irreducible unitary representation of $\Ucal_n$ on $E_\pi$.\\

We are now ready to prove the main theorem
\begin{proof}[Proof of \Cref{theo_entanglementwitness}]
Let $\phi$ be a $t$-positive map with $1\leq t\leq \min\{m,n\}$. Define $\Hcal$ and the $t$-positive map $\Phi: M_n(\Cbb)\to\Bcal(\Hcal)$ 
as in Equations \eqref{eq_H} and \eqref{eq_Phi} respectively. Remark that by the Peter-Weyl Theorem as stated in \Cref{eq_theo_PeterWeyl}, 
$\Hcal=L_2(\Ucal_n,\Cbb^N)$ can be decomposed as
 \[\Hcal=\underset{\pi\in\Lambda}{\widehat {\bigoplus}}\,E_\pi\otimes (F_\pi\otimes\Cbb^N)\,,\]
 and the unitary representation $\lambda\otimes i_N$ of $\Ucal_n$ on $\Hcal$ can be decomposed as
\[\lambda\otimes i_N=\underset{\pi\in\Lambda}{\widehat {\oplus}}\,\lambda^{(\pi)}\otimes i_{F_\pi\otimes\Cbb^N}\,.\]
As $\Lambda$ is a countable set, we can choose an increasing family of finite sets $(\Lambda_d)_{d\geq0}$ such that 
$\cup_{d\geq0}\Lambda_d=\Lambda$. Let $P_d$ be the orthogonal projection on 
$\underset{\pi\in\Lambda_d}{\widehat {\bigoplus}}\,E_\pi\otimes F_\pi$. Remark that $\Hcal_d:=P_d\Hcal$ is a finite 
dimensional subspace of $\Hcal$ as $\Lambda_d$ is finite. We can define on $\Hcal_d$ the restriction of $\Phi$ as 
$\Phi_d:=P_d\,\Phi(\cdot)\,P_d$. This directly defines a $t$-positive map as $\Phi$ is itself $t$-positive by \Cref{lema_Phi}. 
We check that $\Phi_d$ is equivariant. We define the unitary representation $\Pi_d$ on $\Hcal_d$ as the restriction of 
$(\lambda\otimes i_N)$ to $\Hcal_d$: $\Pi_d(U):=P_d\,\lambda(U)\otimes\mathds 1_N\,P_d$. Using that by \Cref{lema_Phi} 
$\Phi$ is $(\lambda\otimes i_N)$-equivariant, it is a straightforward computation to show that $\Phi_d$ is $\Pi_d$-equivariant.\\

Let now $\rho\in\Scal(m\times n)$ be a $t$-entangled density matrix such that $\rho\notin\Scal_m(\phi)$. 
By \Cref{theo_entanglementwitness_infinite}, $\rho\notin\Scal_m(\Phi)$, so that there exists $\psi\in(\Cbb^m\otimes\Hcal)$ such that 
$\langle\psi\,,\,(i_m\otimes\Phi)(\rho)\psi\rangle<0$. Define $\psi_d=P_d\psi\in\Hcal_d$. As $\psi_d\to\psi$ in $\Hcal$ as $d$ 
goes to infinity, we also have $\langle\psi_d\,,\,(i_m\otimes\Phi_d)(\rho)\psi_d\rangle\to\langle\psi\,,\,(i_m\otimes\Phi)(\rho)\psi\rangle$ 
as $d\to+\infty$, 
so there exists $d_0\geq0$ such that for all $d\geq d_0$, $\langle\psi_d\,,\,(i_m\otimes\Phi_d)(\rho)\psi_d\rangle<0$. This concludes the proof.
\end{proof}

Let us finish this section by the following corollary that might be of theoretical interest.
\begin{corollary}\label{coro_entanglementwitness}
There 
exists a family of unitarily equivariant $t$-positive maps $(\Phi_a)_{a\geq1}$ with 
$\Phi_a\,:\,M_n(\Cbb)\to M_{N_a}(\Cbb)$
such that, for each $m$, $(\Scal_m(\Phi_a))_{a\geq1}$ 
decreases to $\operatorname{SN}_t(m,n)$ as $a\to\infty$.
\end{corollary}

\begin{proof}
For the same statement with $m$ fixed, it is a diagonal argument applied to a dense countable subset 
of the $t$-entangled states (given that a direct sum of $l$ entanglement witnesses for $l$ entangled states 
witnesses all $l$ states simultaneously).
For generic $m$ this is an additional diagonal argument. 
\end{proof}

\section{Acknowledgement}
Part of this work was made during IB's visit to Kyoto University, supported by Campus France Sakura project. IB and BC were supported by 
the ANR project StoQ ANR-14-CE25-0003-01. BC was supported by Kakenhi 15KK0162, 17H04823, 17K18734. GS would like to 
acknowledge FRIENDSHIP project of Japan International Corporation Agency (JICA) for research fellowship (D-15-90284) 
and kakenhi 15KK0162. GS would like to thank Prof. B.V Raja Ram Bhat for supporting her from his JC Bose grant.
All authors are very grateful to Prof Hiroyuki Osaka for a careful reading and very helpful comments on a preliminary version of the paper, 
to Prof Yuki Arano for suggesting an improvement to the statement and the proof of Theorem
\ref{theo_char_unitaryequivariant}, and to Prof Masaki Izumi for sharing insight on Theorem
\ref{theo_entanglementwitness}.

\bibliographystyle{ieeetr}
\bibliography{reportmy}

\noindent (Ivan Bardet) Department of Applied Mathematics and Theoretical Physics, University of Cambridge, Cambridge CB3 0WA, England
France \\ E-mail: bardetivan@gmail.com\\

\noindent (Beno\^\i t Collins) Department of Mathematics, Graduate School of Science,
Kyoto University, Kyoto 606-8502, Japan \\
E-mail address: collins@math.kyoto-u.ac.jp\\

\noindent (Gunjan Sapra) Indian Statistical Institute, R. V. College Post, Bangalore-560059, India\\
E-mail address: gunjan18{\textunderscore}vs@isibang.ac.in

	\end{document}